\documentclass[final,12pt]{clear2023} % Anonymized submission
\usepackage{wrapfig}
\DeclareMathOperator{\E}{\mathbb{E}}
\usepackage[format=plain]{caption}

\title[Identifiability and sensitivity analysis of synthetic control models]{Non-parametric identifiability and sensitivity analysis of synthetic control models}

\usepackage{booktabs}
\usepackage{times}

\usepackage{cleveref}
\usepackage{amsmath,amsfonts,amssymb}  
\usepackage{multicol}
\newtheorem{thm}{Theorem}[section]
\newtheorem{define}{Definition}[section]
\clearauthor{%
 \Name{Jakob Zeitler} \Email{mail@jakob-zeitler.de}\\
 \addr University College London\footnote{Research done while author was an intern at Spotify}
 \AND
 \Name{Athanasios Vlontzos} \Email{athanasiosv@spotify.com}\\
 \addr Spotify and Imperial College London%
  \AND
 \Name{Ciaran M. Gilligan-Lee} \Email{ciaran.lee@ucl.ac.uk}\\
 \addr Spotify and University College London%
}

\begin{document}

\maketitle

\begin{abstract}%
Quantifying cause and effect relationships is an important problem in many domains. The gold standard solution is to conduct a randomised controlled trial. However, in many situations such trials cannot be performed. In the absence of such trials, many methods have been devised to quantify the causal impact of an intervention from observational data given certain assumptions. One widely used method are \emph{synthetic control} models. While identifiability of the causal estimand in such models has been obtained from a range of assumptions, it is widely and implicitly assumed that the underlying assumptions are satisfied for all time periods both pre- and post-intervention. This is a strong assumption, as synthetic control models can only be learned in pre-intervention period. In this paper we address this challenge, and prove identifiability can be obtained without the need for this assumption, by showing it follows from the principle of invariant causal mechanisms. Moreover, for the first time, we formulate and study synthetic control models in Pearl's structural causal model framework. Importantly, we provide a general framework for sensitivity analysis of synthetic control causal inference to violations of the assumptions underlying non-parametric identifiability. We end by providing an empirical demonstration of our sensitivity analysis framework on simulated and real data in the widely-used linear synthetic control framework.
\end{abstract}


\section{Introduction}

%**Update the following to include the need for and overview of identifiabity proof. Highlight two main conceptual ingredients of this proof: 1) Donor units are proxies of underlying latent drivers of change, and 2) invariant causal mechanism assumption needed to link pre- and post-intervention periods. This was missing from all previous work. But this requirement is needed for diff-in-diff. Discuss how identifiability proof allows us to derive approach to sensitivity analysis.**

Understanding and quantifying cause and effect relationships is a fundamental problem in numerous domains, from science to medicine and economics---see \cite{gilligan2020causing, richens2020improving, lee2017causal, jeunen2022disentangling, dhir2020integrating, reynaud2022d, gilligan2022leveraging, perov2020multiverse, vlontzos2021estimating}. The generally-accepted gold standard solution to this problem is to conduct a randomised controlled trial, or A/B test. However, in many situations such trials cannot be performed; they could be unethical, exorbitantly expensive, or technologically infeasible. In the absence of such trials, many methods have been developed to infer the causal impact of an intervention or treatment from observational data given certain assumptions. One of the most widely used causal inference approaches in economics \cite{abadie2010synthetic}, marketing \cite{brodersen2015inferring}, and medicine \cite{kreif2016examination} are \emph{synthetic control} methods. 

%The synthetic control method relies on information from both the time evolution of the target unit and other correlated units in order to assess the impact of the intervention. Specifically, it first observes the time-series of the variable of interest before and after the intervention. Further it combines time-series data of other predictor variables correlated with the target variable---but not impacted by the intervention---together with time-series data of the target variable before the intervention to build a \emph{synthetic control} time-series. The method is aimed at approximating the evolution of the target variable in a counterfactual world where the intervention did not occur, all else being equal. By comparing the observed, factual, post-intervention time-series to the counterfactual one, the causal impact of the intervention can be ascertained.

To concretely illustrate synthetic controls, consider the launch of an advertising campaign in a specific geographic region, aimed to increase sales of a product there. To estimate the impact of this campaign, the synthetic control method uses the number of sales of the product in different regions, where no policy change was implemented, to build a model which predicts the pre-campaign sales in the campaign region. This model is then used to predict product sales in the campaign region in the counterfactual world where no advertising campaign was launched. By comparing the model prediction to actual sales in that region after the campaign was launched, one can estimate its impact. 

In the standard synthetic control set-up, the model is taken to be a weighted, linear combination of sales in the no-campaign regions. To train the model, one needs to determine the weights for sales in each no-campaign region that minimise the error when predicting the sales in the campaign region before the campaign was launched. The linearity of the model is justified by assuming an underlying linear factor model for all regions, or units, that is the same for all time periods, both before and after the intervention. Recent work by \cite{shi2022assumptions} has removed the need for the linear factor model assumption and proven identifiability from a non-parametric assumption: that units are aggregates of smaller units. This assumption is reasonable in situations like our advertising campaign example, where total sales in a region is just the aggregate of sales from each individual in that region. However, in many applications, this assumption does not apply. In medicine for instance, patients are not generally considered to be aggregates of smaller units. When the aggregate unit assumption can't be justified, can the causal effect of an intervention on a specific unit be identified from data about ``similar'' units not impacted by the intervention?

Returning to our example, the reason sales in different regions provide good synthetic control candidates is that the causes of sales in most regions are very similar, consisting of demographic factors, socioeconomic status of residents, and so on. Informally, sales in ``similar'' regions act as \emph{proxies} for these, generally unobserved, causes of sales in the campaign region. That is, before the campaign, the causes of sales in the campaign region are also causes of sales in the no-campaign region---they are \emph{common} causes of the campaign and no-campaign regions. This relationship between the target variable and synthetic control candidates is illustrated as a directed acyclic graph, or DAG, in Figure~\ref{fig:dag1}. \cite{shi2021proximal} combined this formulation with results from the proximal causal inference literature to prove one can identify the causal effect of an intervention on the target unit from data about the proxy units not impacted by the intervention. See  \cite{tchetgen2020introduction} for an overview of proximal causal inference. Hence, in our example, observing sales in multiple no-campaign regions allows one to predict the contemporaneous evolution of sales in the campaign region in the absence of the campaign without needing any linearity assumptions.

However, in all previous identifiability proofs, it is implicitly assumed that the underlying assumptions are satisfied for all time periods, both pre- and post-intervention (see assumption 3'' in \cite{shi2021proximal}, and assumption A2 in \cite{shi2022assumptions}). This is a strong assumption, as models can only be learned in pre-intervention period. That is, one of the main assumptions underlying the validity of synthetic control models is that there is no unobserved heterogeneity in the relationship between the target and the control time-series observed in the pre-intervention period. Such unobserved heterogeneity could, for instance, be due to unaccounted-for causes of the target unit. %This assumption is similar to the ignorability assumption encountered in the potential outcomes approach to causal inference, where one assumes all confounders are observed. As such it is important to prove identifiability still holds when this assumption is removed.

In this paper we address this challenge, and prove identifiability can be obtained without the need for the requirement that assumptions hold for all time periods before and after the intervention, by proving it follows from the principle of \emph{invariant causal mechanisms}. Moreover, for the first time, we formulate and study synthetic control models in Pearl's structural causal model framework.

As the assumptions underlying our identifiability proof cannot be empirically tested---as with all causal inference results---it is vital to conduct a formal sensitivity analysis to determine robustness of the causal estimate to violations of these assumptions. In propensity-based causal inference for instance, sensitivity analysis has been conducted to determine how robust propensity-based causal estimates are to the presence of unobserved confounders, see \cite{veitch2020sense} for an overview. These sensitivity analyses derived a relationship between the influence of unobserved confounders and the resulting bias in  causal effect estimation. This understanding allows one to bound bias in causal effect estimation as a function of unobserved confounder influence. From this a domain expert can offer judgments of the bias due to plausible levels of unobserved confounding.

However, despite the importance of this problem---and the wide use of synthetic control methods in many disciplines---general methods for sensitivity analysis of synthetic control methods are under-studied. This work's contributions seek to remedy this discrepancy and provide a general framework for sensitivity analysis of synthetic control causal inference to violations of the assumptions underlying our non-parametric identifiability proof. %The sensitivity to unobserved heterogeneity will be formulated in Pearl's framework of structural causal models \cite{pearl2009causality} using proxy variables, discussed in subsequent sections. 

In summary, our main contributions are as follows:
\begin{enumerate}
\item We formulate  synthetic control models in Pearl's structural causal model framework.
    \item We provide a non-parametric identifiability proof in Pearl's structural causal model framework that doesn't require assumptions to be satisfied before and after the intervention. Our proof relies on the invariant causal mechanism principle.
    \item We provide a general framework for sensitivity analysis of synthetic control causal inference to violations of the assumptions underlying our non-parametric identifiability proof.
    \item We empirically demonstrate our sensitivity analysis approach on real-world data.
\end{enumerate}

\paragraph{Paper Organisation}
As discussed in the introduction, our goal is to identify the causal effect of an intervention, or treatment, on the unit to which it was applied using data from ``similar'' units not impacted by the treatment. First we overview related work, then formulate synthetic control models in Pearl's structural causal model framework, where we prove identifiability using results from proximal causal inference and the assumption that causal mechanisms are invariant. Finally, we provide a formal sensitivity analysis when the assumptions of our identifiability proof fail.

\section{Related work}
\paragraph{\textbf{Identifiability of synthethic controls}} The standard approach to synthetic control models uses the assumption that the data is generated by an underlying linear factor model to derive prove the counterfactual is identified as a linear combination of units not impacted by the treatment, see \cite{abadie2003economic,abadie2010synthetic}. %, and extensions to it have been considered using Bayesian structural time-series \cite{brodersen2015inferring}.
Recent work in \cite{shi2022assumptions} proved that this linearity  emerges in a non-parametric manner if treatment and control units are coarse-grainings of ``smaller'' units, and if causal mechanisms are independent. Recent work by \cite{shi2021proximal}, removed the need for this ``aggregate unit'' assumption, and proved that the counterfactual can be identified as a function of the control units---but this function need not be linear. The result of \cite{shi2021proximal} uses tools from the proximal causal inference literature, see \cite{tchetgen2020introduction} for an overview. Initial proximal causal inference results were reported in \cite{kuroki2014measurement,miao2018identifying}, and have since been developed further and used in long-term causal effect estimation \cite{imbens2022long}. See \cite{shpitser2021proximal} for a formulation of proximal causal inference in the graphical causal inference framework. We note that all the aforementioned works are formulated in the potential outcomes framework for causal inference. Moreover, as mentioned previously, in all these works it is taken that the underlying assumptions are satisfied for all time periods. % (see assumption 3'' in \cite{shi2021proximal}, and assumption A2 in \cite{shi2022assumptions}). 
This is a strong assumption, as models can only be learned in pre-treatment period.

% However, in all of the above identifiability proofs, it is implicitly assumed that the underlying assumptions are satisfied for all time periods (see assumption 3'' in \cite{shi2021proximal}, and assumption A2 in \cite{shi2022assumptions})---both pre- and post-intervention. This is a strong assumption, as models can only be learned in pre-treatment period. Moreover, 

\paragraph{\textbf{Invariant causal mechanisms}} As mentioned, in this paper we prove identifiability can be obtained without the need for this assumption, by showing it follows from the principle of \emph{invariant causal mechanisms}.% Moreover, for the first time, we formulate and study synthetic control models in Pearl's structural causal model framework. 
 Causal mechanisms are invariant if they take the same form in different domains, even though the data distributions may vary with domain. Previous work on invariant causal mechanisms can be found in \cite{mitrovic2020representation,guo2022causal,wang2022out,chevalley2022invariant}. Importantly, this principle is related---yet distinct from---the principle of \textit{independent} causal mechanisms, which says that the mechanism that maps a cause to its effect is independent of the distribution of the cause in a given domain \cite{parascandolo2018learning,stegle2010probabilistic}. In the independent causal mechanism principle, the mechanism itself need not be the same across domains, it just cannot contain information about the distribution of the cause.

\paragraph{\textbf{Sensitivity analsysis}} Later in the paper, we use our identifiability proof to formally investigate synthetic control models from a sensitivity analysis standpoint for the first time. Previous work on sensitivity analysis has investigated omitted variable bias in propensity-based models. This sensitivity analysis work originated in \cite{imbens2003sensitivity,rosenbaum1983assessing} with modern extensions in \cite{veitch2020sense} and \cite{cinelli2020making,cinelli2019sensitivity}.

\section{Methods}
\subsection{Non-parametric identifiability from proxies and invariant causal mechanisms} \label{section: Pearl SCM}

\begin{figure}[t]
\centering
\subfigure[
    \label{fig:dag1}]{
        \includegraphics[scale=0.35]{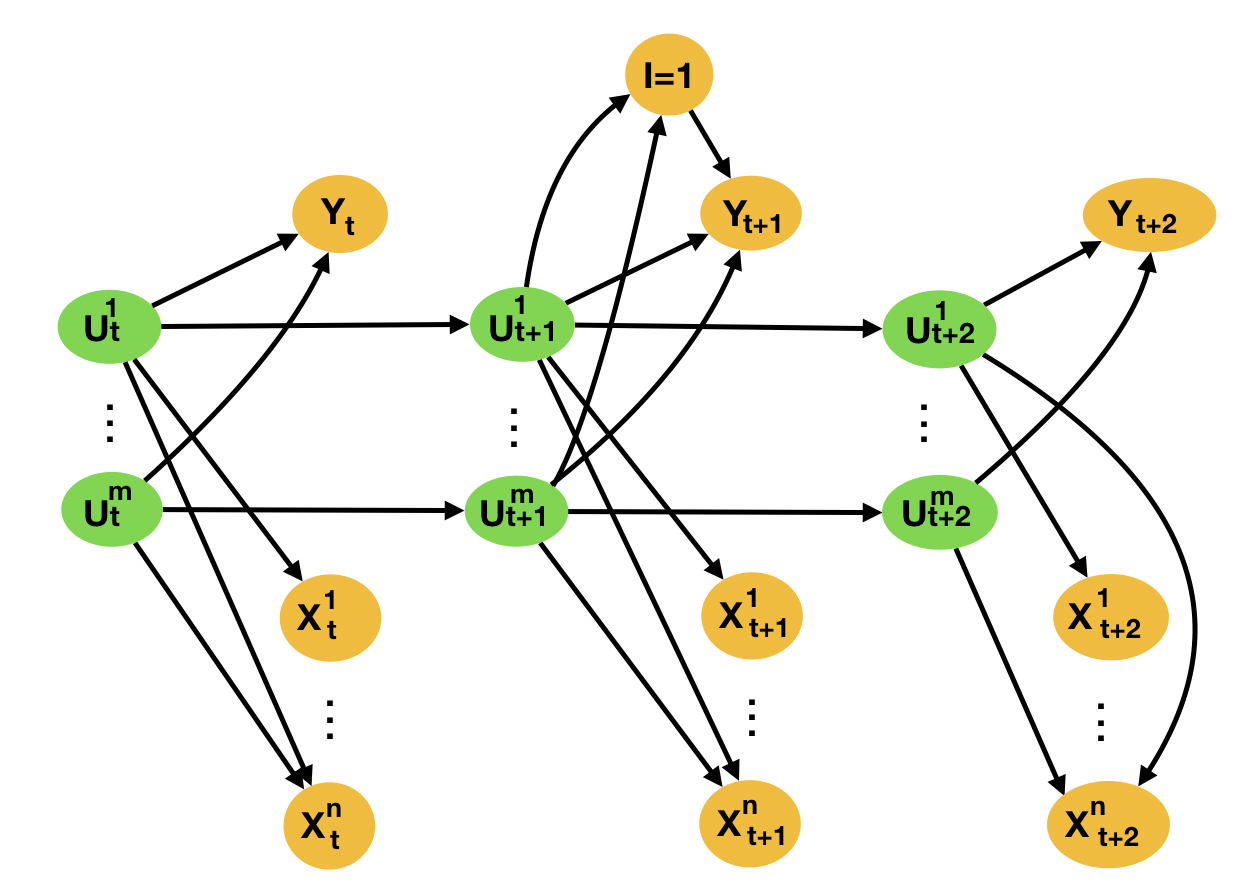}
 }\quad
    \subfigure[
    \label{fig:twin net}]{
        \includegraphics[scale=0.35]{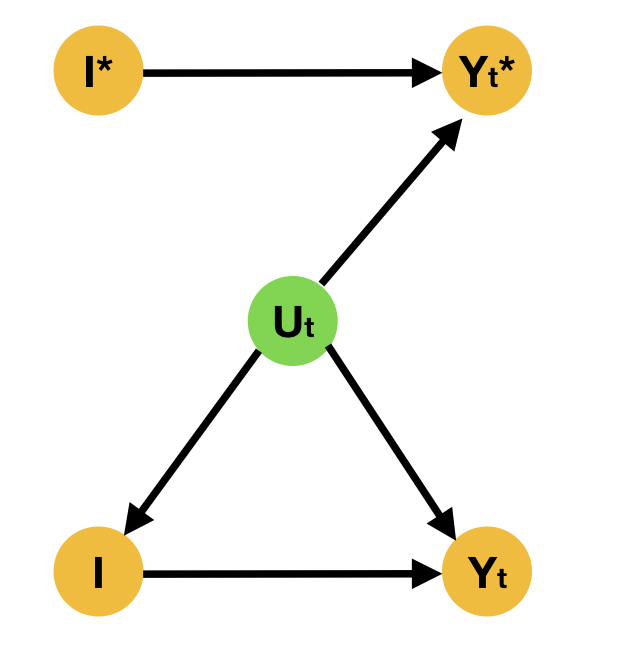}
        }
        \caption{(a) DAG for synthetic control method. Orange nodes are observed variables, green latent. Intervention is applied at timepoint $t+1$. For simplicity, $I$ is taken to be $0$ for all time points before $t+1$, and $1$ for all timepoints from $t+1$ on-wards. (b) Twin network for target unit. Superscript $^*$ denotes counterfactual world.}
    \end{figure}
    
    %Our proof has two main conceptual ingredients:
    %\begin{itemize}
        %\item Donor units are proxies of underlying latent causes of the treated unit
        %\item invariant causal mechanism assumption needed to link pre- and post-intervention periods
    %\end{itemize}
    We work in the structural causal model framework of \cite{pearl2009causality}. We now will present our definition of a synthetic control structural causal model, define invariant causal mechanisms, and formally define proxy variables following the proximal causal inference literature of \cite{tchetgen2020introduction}. 
    
    \begin{define}[Synthetic Control Structural Causal Model   ]
    A synthetic control structural causal model consists of \textbf{a set of latent variables $U$} and their distributions, \textbf{a set of observed variables $\mathbf{Y,X, I}$} representing the target unit, donor units, and the intervention, and a \textbf{set of deterministic functions} mapping parents to their children in the causal structure in Figure~\ref{fig:dag1}(a), represented as a directed acyclic graph (DAG), each indexed by a specific time point $t$, such that:
       \begin{enumerate}
           \item $u^t=m^t(u^{t-1}, \epsilon^t_{u})$ where $\epsilon^t_{u}$ is an independent, exogenous error term with $\epsilon^t_{u}\sim P(\epsilon^t_{u})$ 
           \item $y^t=g^t(u^t, I^t, \epsilon^t_{y})$ where $\epsilon^t_{u}$ is an independent, exogenous error term with $\epsilon^t_{y}\sim P(\epsilon^t_{y})$ 
           \item $x^t_i=f^t(u^t, \epsilon^t_{x_i})$ where $\epsilon^t_{x_i}$ is an independent, exogenous error term with $\epsilon^t_{u}\sim P(\epsilon^t_{x_i})$ 
       \end{enumerate}
       For simplicity, we follow \cite{zhang2022can} and suppress the functional dependence on the exogenous error terms.  
    \end{define}
    
    The above formulation in terms of structural causal models generalises the standard formulation of synthetic controls in terms of linear factor models. For instance, $y^t$ is considered an arbitrary function of $u^t, I^t$, rather than a linear function of them.  In what follows we treat $x^t, y^t$ as the variables we are concerned with. Sometimes we abuse notation and use the same $x^t, y^t$ to denote the values those variables take. The difference will be clear from the context.
    
    %**Define causal mechanisms more precisely and show how they relate to conditional distributions of the variable in question conditioned on all it's parents, both exogenous and endogenous. Invariance of causal mechanisms means that they no longer depend on time** 
    
    The collection of functions and distribution over latent variables induces a distribution over observable variables: $P^t(x^t) = \sum_{u^t} P^t(x^t \mid u^t)P^t(u^t) = \sum_{u^t} \delta_{f^t(u^t)=x^t}P^t(u^t). $ Where $\delta_{f^t(u^t)=x^t} = 1$ when $u^t$ satisfies $f^t(u^t)=x^t$, and $0$ otherwise. For any variable in a causal model, its \emph{causal mechanism} is the deterministic function that determines it from its parents in the causal structure. This function is equivalent to the conditional distribution of that variable given it's parents. For instance:
    $ x^t=f^t(u^t) \iff P^t(x^t \mid u^t)=1.$
    
     \begin{define}[Invariant causal mechanisms]\label{definition: inv causal mech}
       In the context of synthetic control causal models, a causal mechanism is said to be \emph{\textbf{invariant}} if it doesn't depend on the time point $t$.
    \end{define}
    
    The structural causal model framework allows us to define (strong) interventions via the \emph{do-operator}, which replaces the original causal mechanism with assignment of that variable to a specific value, disconnecting the intervened variable from its parents in the causal structure \cite{pearl2009causality}.
    
    To formally define when a collection of variables are to be considered proxies for other variables, we need the following completeness condition.

    \begin{define}[Completeness condition for proxy variables]\label{definition: completeness}
       For any square integral function $f$, if $\mathbb{E}\left(f(x^t_1, \dots, x^t_N) \mid u^t \right) = 0$ then $f(x^t_1, \dots, x^t_N)=0$ for any $t$.
    \end{define}

This completeness condition characterizes how much ``information'' the $x^t$ have about the $u^t$, in the sense that $x^t$ have sufficient variability relative to the $u^t$---that is, any variation in $u^t$ is captured by variation in $x^t$. Such completeness conditions are widely assumed in recent proximal causal inference literature \cite{tchetgen2020introduction}, and under these conditions the $x^t$ can be viewed as proxy variables for the latents $u^t$. 

To quantify the impact of an intervention $I=1$ on unit $y$ at time $t$, we must estimate the effect of treatment on the treated:
$$\underbrace{\mathbb{E}^t\left(y^t \mid \text{do}(I^t=1), I^t=1 \right)}_{\text{Observed}} -  \mathbb{E}^t\left(y^t \mid \text{do}(I^t=0), I^t=1 \right)$$
As we observe the first term, all that is required is to identify and estimate the second term.
    
The below Theorem~\ref{theorem: for a single time point} and proof is based on Theorem~4 in \cite{shi2021theory}. The main difference is our assumptions and the causal framework we work in. We work in Pearl's graphical causal model framework, where independence of causal variables follow from graphical conditions in the given causal structure, represented as a DAG. Indeed, even conditional independence of counterfactual variables can be seen to follow from graphical requirements---this time by considering the structure of the \emph{twin network} associated with the causal structure, see \cite{vlontzos2021estimating,graham2019copy} for an overview of twin networks. \cite{shi2021theory} work in the potential outcomes framework, and thus require explicit assumptions for various conditional independence statements. Additionally, Theorem~4 in \cite{shi2021theory} assumed the existence of a function that maps the control units to the target unit that is consistent and unchanged across all time points. We do not make this assumption. Rather in our Theorem~\ref{theorem: singe h for all time points} we remove this assumption and prove that such a function\footnote{Existence of such a function for a single timepoint follows from proximal causal inference, as shown in Theorem~\ref{theorem: for a single time point}.} is the same for all time points if causal mechanisms are invariant. 

For simplicity, we will denote the collection of donor units at timepoint $t$, $\{x^t_1, \dots, x^t_N\}$, by $x^t.$
    
\begin{thm}\label{theorem: for a single time point}
    There exists a function $h^t$ such that at time point $t$ we have $$\mathbb{E}\left(y^t \mid \text{do}(I^t=0), I^t=1 \right) = \mathbb{E}\left(h^t(x^t, I^t=0) \right),$$
    where $$\mathbb{E}\left(y^t \mid \text{do}(I^t=0), I^t=1 \right) = \int y^t P^t\left(y^t \mid \text{do}(I^t=0), I^t=1\right)dy,$$ and $$\mathbb{E}\left(h^t(x^t, I^t=0) \right) = \int h^t(x^t, I^t=0) P^t\left(x^t\right)dx.$$
    \end{thm}
    
    \begin{proof} 
    In general nonparametric models, the completeness condition of Def.~\ref{definition: completeness} together with some additional technical conditions (see the appendix for these technical conditions) imply the existence of a function\footnote{To gain some intuition about the existence of such functions, a simple example is: $P(Y) = \int P(Y \mid X) P(X)dx = \int H(Y, X) P(X)dx.$} $H^t$ such that
    \begin{equation}\label{equation: existence of H}
        \begin{aligned}
        P^t(y^t \mid u^t, I) &=
        \int H^t(y^t, x^t, I^t) P^t(x^t \mid u^t, I^t)dx 
        \end{aligned}
    \end{equation}
This implies that
 \begin{equation*}
        \begin{aligned}
       \mathbb{E}\left(y^t \mid u^t, I^t \right) &= \iint y^t H^t(y^t,x^t, I^t) P^t(x^t \mid u^t, I^t)dydx  \\
       &\overset{A}{=} \int{P^t(x^t \mid u^t)}\underbrace{\left[\int{y^tH^t(y^t, x^t, I^t)}dy\right]}_{:=h^t(x^t, I^t)}dx\\
       &=\int h^t(x^t, I^t)P^t(x^t \mid u^t, I^t)dx \\
       &\overset{A}{=}\int h^t(x^t, I^t)P^t(x^t \mid u^t, I^t)dx \\
       &\overset{B}{=}\mathbb{E}\left(h^t(x^t, I^t) \mid u^t, I^t \right).
        \end{aligned}
    \end{equation*}
Where on lines $\text{A}$ we use the fact that $x^t$ is independent of $I^t$ conditioned on $u^t$.

Now, consider the following:
   \begin{equation*}
        \begin{aligned}
       \mathbb{E}\left(y^t \mid u^t, I^t=1 , \text{do}(I^t=0)\right) &\overset{C}{=} \mathbb{E}\left(y^t \mid u^t,I^t=0\right) \\
       &\overset{D}{=}\mathbb{E}\left(h^t(x^t, I^t=0) \mid u^t, I^t=0 \right) \\
    &\overset{E}{=}\mathbb{E}\left(h^t(x^t, I^t=0) \mid u^t \right)
        \end{aligned}
    \end{equation*}
    Line $\text{C}$ in the above follows from examining the twin network in Figure~\ref{fig:twin net} and applying d-separation. Line $\text{D}$ is just the application of line $\text{B}$, above. Line $\text{E}$ follows as $x^t$ is independent of $I^t$ given $u^t$.  Marginalising over $u^t$ yields:
    $\mathbb{E}\left(y^t \mid \text{do}(I^t=0), I^t=1 \right) = \mathbb{E}\left(h^t(x^t, I^t=0) \right)$
  \end{proof} 
  
Theorem~\ref{theorem: for a single time point} proved existence of a function mapping control units to the target unit at a given time. We now prove this function is the same for all timepoints if causal mechanisms are invariant.
  
  \begin{thm}\label{theorem: singe h for all time points}
   If causal mechanisms are invariant, then there exists a \emph{unique} function $h$, such that for \emph{\textbf{all}} time points $t$ we have: $$\mathbb{E}\left(y^t \mid \text{do}(I^t=0), I^t=1 \right) = \mathbb{E}\left(h(x^t, I^t=0) \right)$$
  \end{thm}
  \begin{proof}
   All we need to show is that the solution to the integral equation from Eq.~\ref{equation: existence of H} in the proof of Theorem~\ref{theorem: for a single time point} for time point $t$ is also a solution for any other time point $t'$. 
 
 To show this, first reconsider Eq.~\ref{equation: existence of H}:
   $$
        P^t(y^t \mid u^t, I^t) =
        \int H^t(y^t, x^t, I^t) P^t(x^t \mid u^t, I^t)dx. 
    $$
    Consider the left hand side $P^t(y^t \mid u^t, I^t).$ This is the causal mechanism for determining $y^t$. As causal mechanisms are invariant, this means $P^t(y^t \mid u^t, I^t) = P^{t'}(y^{t} \mid u^{t}, I^{t}).$ Moreover, considering the right hand side of the above Eq.~\ref{equation: existence of H}, as $x^t$ is independent of $I^t$ conditioned on $u^t$:
     $P^t(x^t \mid u^t, I^t) = P^{t}(x^t \mid u^t),$
    which is the causal mechanism for determining $x^t$. Again, as causal mechanisms are invariant, one has that  $ P^t(x^t \mid u_t) = P^{t'}(x^{t} \mid u^{t}).$
    
    Hence a solution to the integral equation for one time point $t$, is a solution for any other time point $t'$. All that remains is to prove uniqueness of $h^t$ for a given time point, as this will imply there exists a unique function $h$ for all time points via the above argument. Suppose there are two functions that are each solutions to line $\text{E}$:
    $\mathbb{E}\left(h^t(x^t, I^t) \mid u^t, I^t \right) = \mathbb{E}\left(\widetilde{h^t}(x^t, I^t=1) \mid u^t, I^t=1 \right)$
    As $x^t$ is independent of $I^t$ conditioned on $u^t$ we have:
    $\mathbb{E}\left(h^t(x^t, I^t=1) - \widetilde{h^t}(x^t, I^t=1) \large\mid u^t \right) = 0$
    
    As this is the expectation of a function of $x^t$ conditioned on $u^t$, the completeness condition in Definition~\ref{definition: completeness} implies that $h^t(x^t, I^t) = \widetilde{h^t}(x^t, I^t),$ completing the proof.
  \end{proof}
 In the standard synthetic control case, $h$ is a linear function of the proxies, as in \cite{abadie2010synthetic}.

\subsection{Sensitivity analysis and bias when identifiability fails}

    \begin{figure}[t]
        \centering
        \includegraphics[scale=0.35]{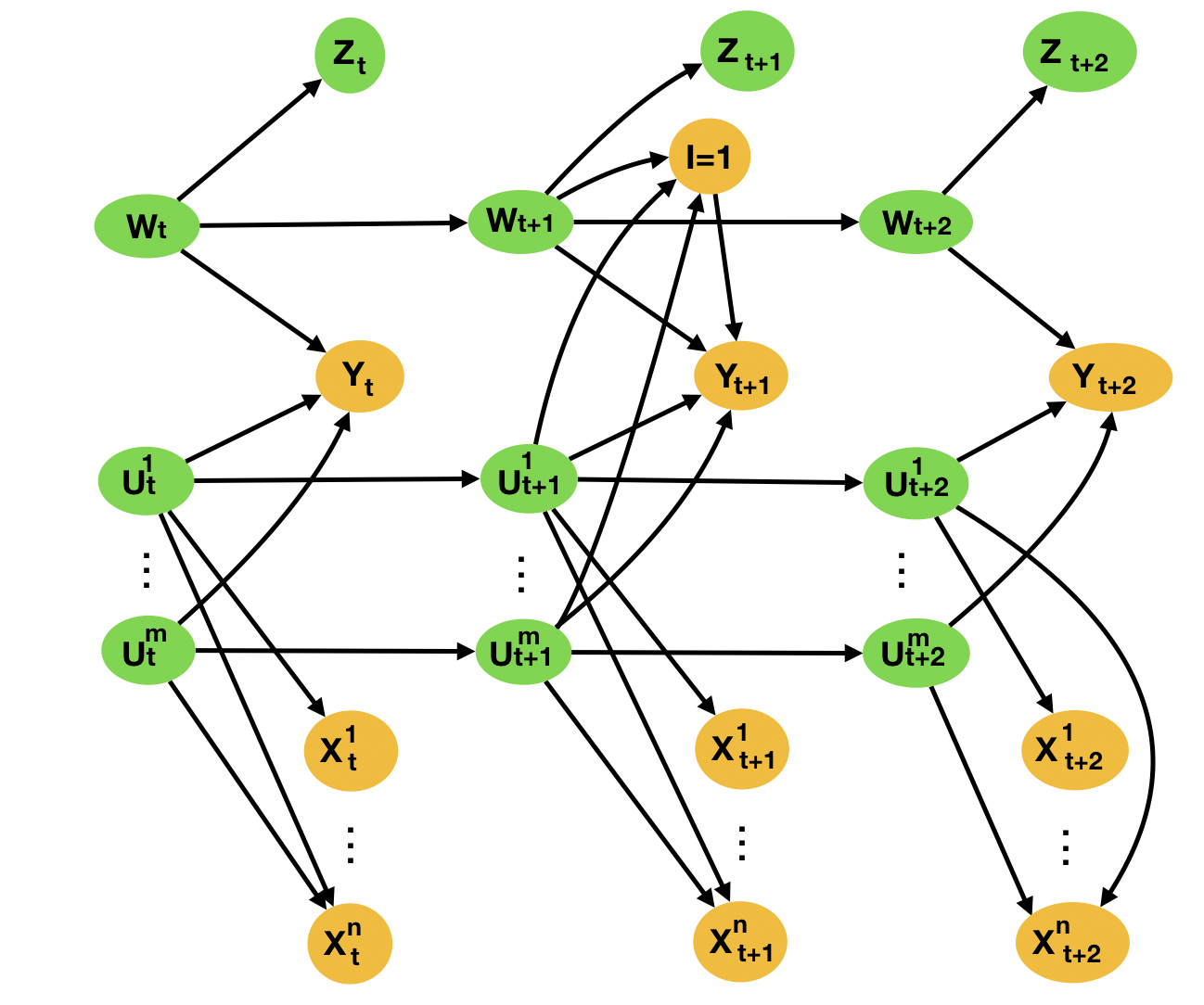}
        \caption{DAG for synthetic control model when there are latent causes of the treated unit that we don't have observed proxies for. Orange nodes are observed, green nodes are latent.}
        \label{fig:dag_missing}
    \end{figure}
    
    If there is a latent cause $w^t$ with unobserved proxies $z^t$, as graphically illustrated in Figure~\ref{fig:dag_missing}, this impacts our identification strategy. In this situation, the updated argument of Theorem~\ref{theorem: for a single time point} proceeds as follows. There exists a function $H$ such that $
        P(y^t \mid u^t, I^t, w^t)=\iint H(y^t, x^t, z^t, I^t) P(x^t, z^t \mid u^t, I^t, w^t)dxdz.$ This implies that:
 \begin{equation*}
        \begin{aligned}
       \mathbb{E}\left(y^t \mid u^t, I^t \right) &= \iiiint y^t H(y^t,x^t, I^t, z^t) P(x^t, z^t \mid u^t, I^t, w^t)P(w^t)dydxdzdw \\
       &\overset{A}{=} \iiiint y^t H(y^t,x^t, I^t, z^t) P(x^t \mid u_t)P(z^t \mid w^t)P(w^t)dydxdzdw \\
       &\overset{A}{=} \int{P(x^t \mid u^t)}\int\underbrace{\left[\int{y^tH(y^t, x^t, I^t, z^t)}dy\right]}_{:=g(x^t, I^t, z^t)}P(z_t)dzdx\\
       &=\int \mathbb{E}_{P(z^t)}\left(g(x^t, I^t, z^t)\right)P(x^t \mid u^t,  I^t)dx \\
       &=\mathbb{E}\left(\underbrace{\mathbb{E}_{P(z^t)}\left(g(x^t, I^t, w^t)\right)}_{:=h(x^t, I^t, P(z^t))} \mid u^t, I^t \right).
        \end{aligned}
    \end{equation*}
    Where lines $\text{A}$ follow as $x^t,z^t$ are independent of $I$, and $w^t, u^t$ respectively. Yielding the Theorem:
    
    \begin{thm}\label{theorem: sensitivity}
     The introduction of latent $w^t$ with no observable proxies, as graphically depicted in Figure~\ref{fig:dag_missing}, changes the estimand in Theorem~\ref{theorem: for a single time point} to:
    $$\mathbb{E}\left(y^t \mid \text{do}(I^t=0), I^t=1 \right) = \mathbb{E}\left(h(x^t, I^t=0, P(z^t)) \right).$$
    \end{thm}
Hence our estimation of the counterfactual depends on the distribution $P(z^t)$. If this distribution is the same for all time periods, our estimation should proceed in an unbiased fashion. However, if there is a distribution shift in $P(z^t)$ between the pre-intervention period (time points for which $I=0$) and the post-intervention period (time points for which $I=1$) then our estimate for the effect of treatment on the treated could be biased. Why is this? Well when we learn the function $h$, we only have access to pre-intervention data. Hence, if the latent cause---and thus the unobserved proxies---undergo a distribution shift, this can bias our model, as the $h$ we learn depends on the distribution of the proxies at the time at which we learnt it, which is before the intervention was applied. The bias is thus given by: 
    \begin{equation} \label{equation: bias}
    \text{Bias} = \big\vert \mathbb{E}\left(h(x^t, I^t=0, P_{pre}(z^t)) \right) - \mathbb{E}\left(h(x^t, I^t=0, P_{post}(z^t)) \right) \big\vert
    \end{equation}

%Bias is the difference between the estimated causal effect given the synthetic control candidates, and the ground truth causal effect. Such bias can arise when the relationship between observed synthetic control candidates, $\{x^i\}$, and target variable, $y$, in the pre-intervention period does not continue in the post-intervention period. For instance, the observed synthetic control candidates may not capture all drivers of $y$, such as $W$ in figure~\ref{fig:dag}, and the distribution of $W$ could shift from the pre-intervention period to the post-intervention period. Thus, $P(\Tilde{y}_{t:m+1}| y_{m:1}, x^1_{t:1}, \dots, x^n_{t:1})$ cannot tell us about heterogeneity in $W$ in the post-intervention period. To perform sensitivity analysis, we will aim to simulate such heterogeneity by subjecting $P$ to a distributional shift. %Indeed, any heterogeneity in $W$, if observed, would have lead to a distributional shift in $P$, taking us to some distribution $Q$. 

%For sensitivity analysis, we need to understand how specific levels of unobserved heterogeneity lead to bias. We take the distance between $P$ and a hypothetical $Q$ to quantify the unobserved heterogeneity and proceed to compute the bias given specific levels of unobserved heterogeneity. Given a way to compute worst-case bias for a given distance, we can plot how bias varies as a function of ``strength of heterogeneity,'' the analogue of the plots used in propensity-based causal inference. Based on this discussion, we formulate the worst-case bias for a given level of heterogeneity as follows.
    
    \subsection{Bounding bias in standard linear synthetic control models}
    
    When we consider the case where $h$ is a linear function of the proxies---the standard case employed in previous works, see \cite{abadie2010synthetic}---the bias takes on a simpler form. That is, if the $z^t$ are the unobserved proxies of the latent $w^t$, as shown in the DAG in Figure~\ref{fig:dag_missing}, then we can write:
    
    $$
    \begin{aligned}
    h(x^t, I=0, P(z^t))  &= \mathbb{E}_{P(z^t)}\left(\sum_i^N \beta_i x^t_i + \sum_j^M\gamma_j z_j^t\right) = \sum_i^N \beta_i x^t_i + \sum_j^M\gamma_j\mathbb{E}_{P(z^t_j)}\left(z^t_j\right)
    \end{aligned}
    $$
    In this case, Theorem~\ref{theorem: sensitivity} implies that the bias can be written as:
    $$
    \begin{aligned}
    \text{Bias} &= \big\vert \mathbb{E}\left(h(x^t, I=0, P_{pre}(z_t)) \right) - \mathbb{E}\left(h(x^t, I=0, P_{post}(z_t)) \right)\big\vert \\
    &=\bigg\vert \sum_j^M\gamma_j\mathbb{E}_{P_{pre}(z^t_j)}\left(z^t_j\right) - \sum_j^M\gamma_j\mathbb{E}_{P_{post}(z^t_j)}\left(z^t_j\right)\bigg\vert \\
    &\leq M \max_j(\vert \gamma_j\vert) \max_j\left(\vert \mathbb{E}_{P_{pre}(z^t_j)}\left(z^t_j\right) - \mathbb{E}_{P_{post}(z^t_j)}\left(z^t_j\right) \vert\right)
     \end{aligned}
    $$
    
As this bound on the bias is in terms of latent quantities, an analyst will need to make \emph{plausibility judgments} in order to devise a bound in terms of observable quantities. Indeed, if an analyst believes they have not missed latent causes as important to our problem as the ones they included proxies for, then we can upper bound the bias in the worst case by taking the maximums in the above bound on the bias to be the maximums in the \textit{observed} proxies. This then leads to the following upper bound on the bias in terms of observable quantities:
    \begin{equation} \label{equation: linear bias bound}
    \begin{aligned}
    \text{Bias} \leq N \times \max_i(\vert \beta_i\vert)\times \max_i\left(\vert
    \mathbb{E}\left(x^{pre}_i\right) - \mathbb{E}\left(x^{post}_i\right) \vert\right)
     \end{aligned}
    \end{equation}

\section{Experiments}
We now assess the validity of this bound on a series of synthetic and real world data. 
Using simulations, we investigate our bound in a valid and invalid setting. Moreover, we test on the California Tobacco Tax and German Reunification data-sets to demonstrate the bound in a real world setting.

\subsection{Synthetic Experiments}

Our synthetic experiments are constructed such that the unobserved latent $w$ experiences a distribution shift after the intervention, leading to bias as defined in Equation~\ref{equation: bias}. To test validity of our bound in Equation~\ref{equation: linear bias bound}, we consider two examples: one where the plausibility bounds are satisfied, illustrated in Figure~\ref{fig:simulation}(a), and one where they are violated, illustrated in Figure~\ref{fig:simulation}(b). Data generation is outlined in Appendix~\ref{section: appendix simulation exp}. Results and discussion are in the caption of Figure~\ref{fig:simulation}.

\begin{figure}  
\centering
\subfigure[
    \label{fig:first}]{
        \includegraphics[width=0.4\linewidth]{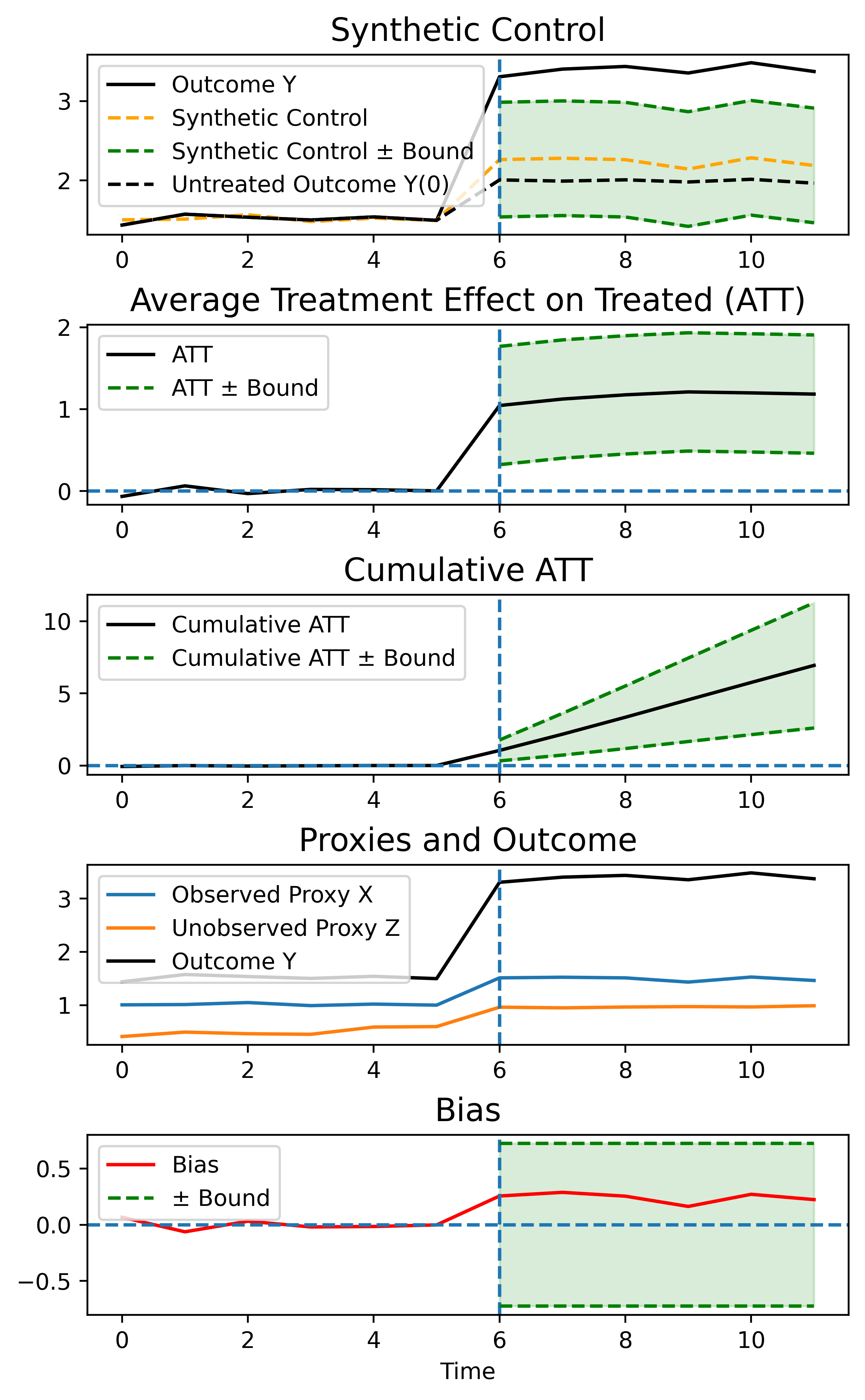}
        }
        \subfigure[
    \label{fig:second}]{
        \includegraphics[width=0.4\linewidth]{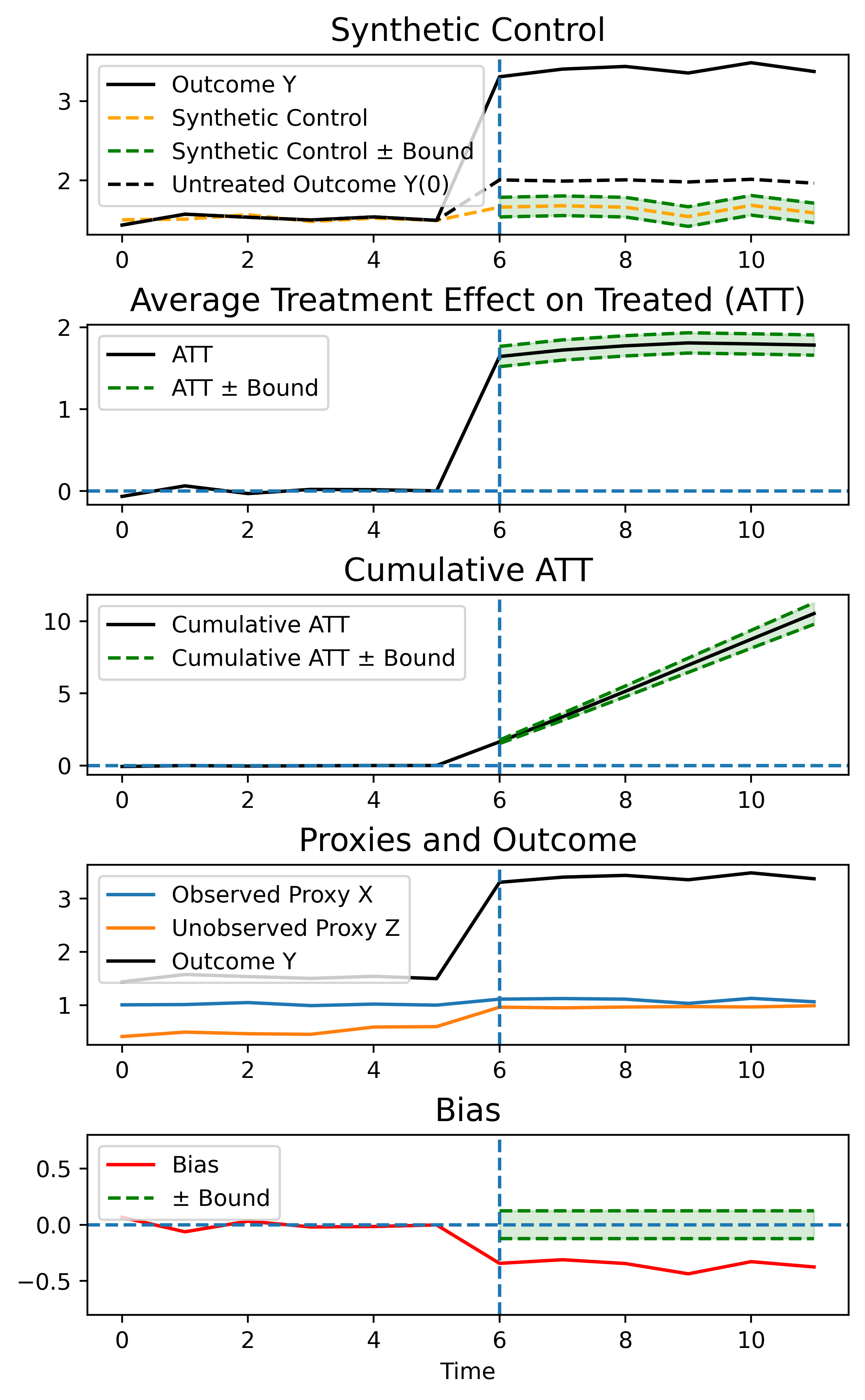}
        }

  \caption{\textbf{Simulation to evaluate the validity of bound} For Scenario (a) we have a \textit{valid} bound and for Scenario (b) an \textit{invalid} bound on the bias, i.e. the bounds do not capture the true bias (red, bottom row). The intervention occurs at $T=6$. The first row shows the outcome of interest, as well as its untreated state $Y \mid \text{do}(I=0)$, written as $Y(0)$ in the figure and depicted in black. The synthetic control is shown as a dashed orange line contained by the bounds in green. The second row shows the average treatment effect on the treated (ATT) by subtracting the synthetic control from the observed outcome and averaging. The third row shows the cumulative effect over time. Row four shows the progression of all proxies, observed and unobserved as well as the outcome over time. The last row shows the bias as defined by the true untreated outcome subtracted from the synthetic control. 
  \textbf{(a)}: The bias (red) is contained by the bounds. Given the bounds do not contain 0 in the ATT plot, the effect measured is still positive even if we had a worst case bias, given our assumptions. Thus, when our plausibility assumptions are satisfied, so too is our bound on the bias.
  \textbf{(b)}: The observed proxy $X$ shifts less during the intervention time (blue) such that the bound is smaller as we measure a smaller change in the proxies. As a consequence, the bias (red) is outside the bounds. Hence, when our plausibility assumptions are violated, so too is our bound on the bias. %(green which is expected as the second condition is violated due to the mean change in the observed proxies $X$ being smaller than the mean change in the unobserved proxies $Z$ of the latent $w$).
 }
  \label{fig:simulation}
\end{figure}

\subsection{Real Data}
In our first experiment on real data, we look into a tobacco tax increase of 25 cents introduced in California in 1988 \cite{abadie2010synthetic}. We build a synthetic control to predict the untreated annual per-capita cigarette sales of California, using sales data from the states used in the literature: Colorado,  Connecticut, Montana, Nevada, Utah. For our second experiment we refer to the 1990 reunification of West and East Germany  in~\cite{abadie2015comparative}. Here, we build a synthetic control model to predict the untreated GDP of West Germany using GDP data from the countries used in the literature: USA, Austria, Netherlands, Switzerland, Japan. We run a linear regression for each synthetic control, without intercept and allowing for negative coefficients. %To account for the linear trend in both time-series, the change in mean for the proxies is calculated as the difference between the mean of the first and second half of the pre-intervention stage, with the same procedure applied to the post-intervention stage. We further limit the mean change calculation to 2 time-steps, in each the pre and post-intervention stage. 
In line with \cite{brodersen2015inferring}, we calculate the ATT as a running average in $T_{post}$.

\begin{figure}  
\centering
\subfigure[
    \label{fig:real_first}]{
        \includegraphics[width=0.4\linewidth]{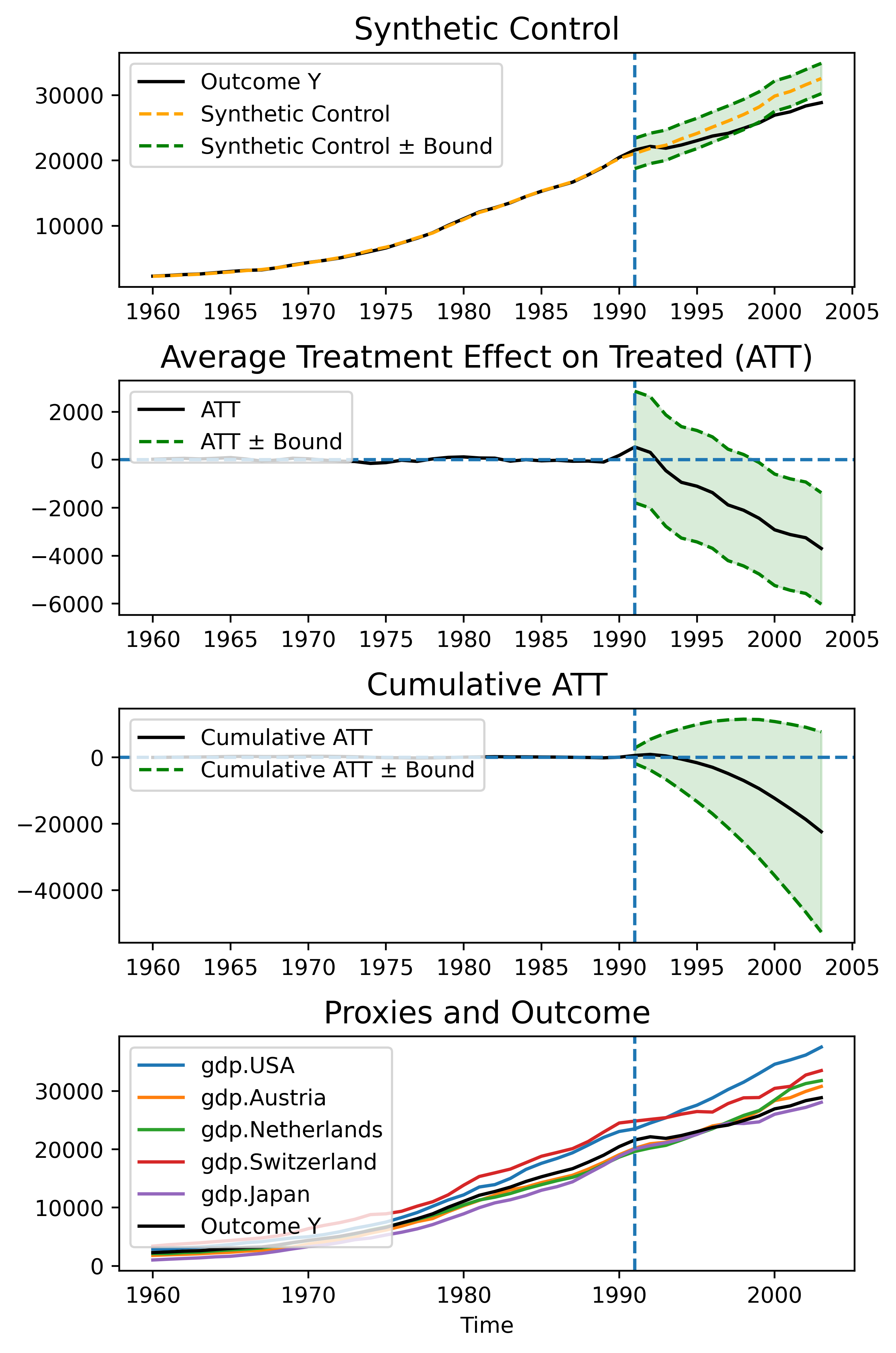}
        }
        \subfigure[
    \label{fig:real_second}]{
        \includegraphics[width=0.4\linewidth]{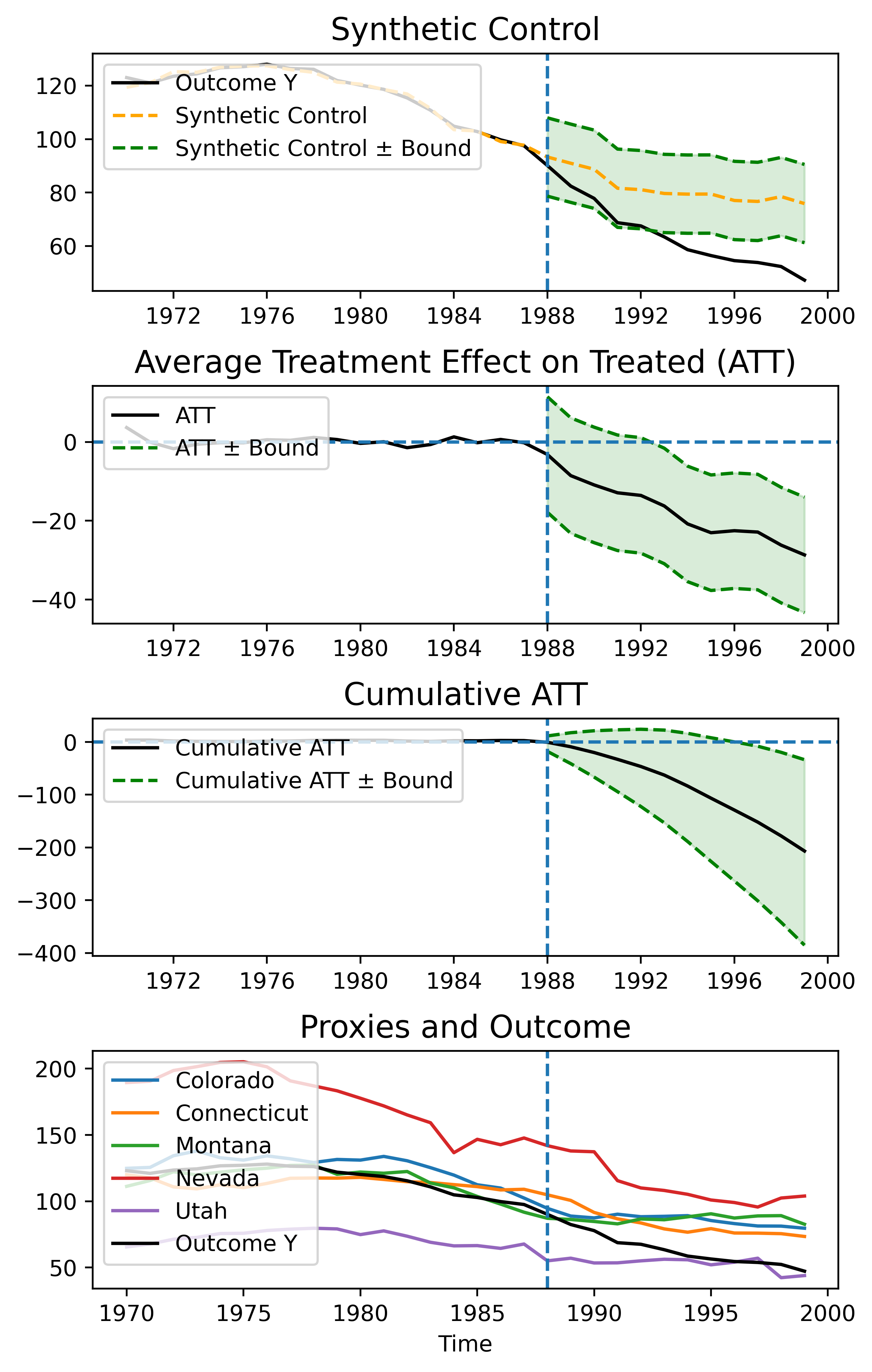}
        }

  \caption{\textbf{Real World Examples} 
  \textbf{(a)}: German Reunification: The reunification occurred in 1991 leaving West Germany with a visible drop in its previous growth. Our bound of 2321.84 does not allow to conclude that, given our assumptions, the average ATT is not reversed by other latent variable shifts by $u$. The bound on the bias is primarily driven by the high levels of variation in Switzerland (red, bottom plot) which around the time of the reunification also changed into a phase of slower growth.
  \textbf{(b)}: California Tobacco Tax: Our bound of 14.65 is smaller than the average ATT of -17.45 and allows us to conlcude, that given the assumptions, that even if there was an unobserved bias as big in contribution as the observed proxies, the negative causal effect would persist. Based on that reasoning, the California Tobacco Tax can be deemed as effective in reducing smoking.
 }
  \label{fig:examples}
\end{figure}

Table \ref{table:1} shows the bounds our method yields. For German Reunification, we have $N=4$ as Japan's coefficient is zeroed. The biggest beta coefficient corresponds to Austria with 0.46. The proxy change is 1252, yielding a bound of 2321.84. With an average ATT of -1726.8, given our assumptions, this bound on the bias tell us that the causal effect we have estimated is very sensitive. This is, in the worst case, the causal estimate in this case can be entirely due to a shift in an unobserved latent. For California Tobacco, we have $N=4$ as the regression zeroes the coefficient on Utah, leaving 4 proxies. Montana has the biggest regression coefficient with 0.4. The maximum change in the proxies is 9.1, yielding a bound of 14.65, which is smaller than the average ATT of -17.45. In contrast to the German Reunification example, our bound in this example allows us to conclude that---given our assumptions---even with the worst case bias, the tobacco tax will still have a negative causal effect. See Figure \ref{fig:examples} for the corresponding synthetic control plots.

\begin{table}[]
\begin{tabular}{|l|l|l|l|l|l|}
\hline
\textbf{Data}   & \textbf{N} & \textbf{Max. Beta} & \textbf{Max. Proxy Change} & \textbf{Bound} & \textbf{average ATT} \\ \hline
Germany, 1990    & 4          & 0.46              &  1252                        & $\pm$2321.84. & -1726.80        \\ \hline
California, 1981 & 4          & 0.4               & 9.1                    & $\pm$ 14.65     & -17.45     \\ \hline
\end{tabular}
\caption{Bounds on bias of ATT for the California Tax and German Reunification Example}
\label{table:1}
\end{table}

By design, sensitivity analysis is a subjective method, as it relies on domain expert knowledge to make a judgement on the empirical evidence given. Our method offers a conservative upper bound on the bias, where both the maximum beta and proxy change are empirical, and the domain expert is left with the decision which proxies to incorporate into their analysis. Effectively, this is  equivalent to the unobserved parameters commonly introduced to models in classical sensitivity analysis \cite{imbens2003sensitivity}, on which the expert has to make their judgement. Here, our aim was not a final judgement on the real world examples shown, but instead to demonstrate of how to enrich expert discussion with our bound for any synthetic control analysis. Ultimately, it is the expert that has to make \textit{plausibility judgements} in the scientific discourse, and these bounds are a necessary addition to understand robustness against bias.

%\begin{figure}[t]
%        \centering
%        \includegraphics[scale=0.35]{example1a.png}
%        \caption{Example setup: one unobserved latent $\lambda$, with two observed proxies $W_1, W_2$, one latent $U$ which together with $\lambda$ constitutes the two parents of outcome $Y$}
%        \label{fig:dag1}
%    \end{figure}

\section{Conclusion}
One of the most widely used causal inference approaches are \emph{synthetic control} methods. However, in all previous identifiability proofs, it is implicitly assumed that the underlying assumptions are satisfied for all time periods both pre- and post-intervention. This is a strong assumption, as models can only be learned in pre-intervention period. In this paper we addressed this challenge, and proved identifiability without the need for this assumption by showing it follows from the principle of invariant causal mechanisms. Moreover, for the first time, we formulated and studied synthetic control models in Pearl's structural causal model framework. Importantly, we provided a general framework for sensitivity analysis of synthetic control models to violations of the assumptions underlying non-parametric identifiability. We concluded by providing an empirical demonstration of our sensitivity analysis approach on real-world data.

\bibliography{references.bib}
%\newpage

\appendix

\section{Technical conditions for proofs}
The proof of Theorem~\ref{theorem: for a single time point} relies on some technical conditions which we now overview. See Appendix C of \cite{shi2021proximal} and references therein for further details. In order to prove the existence of the function $H^t$, we need the following.

Consider the space of all square-integrable functions $s$, denoted $L^2\{F(s)\}$, with respect to a cumulative distribution function $F(s)$. This is a Hilbert space with inner product given by $\langle f,g \rangle = \int f(s)g(s)dF(s)$. Let $K_x$ denote the conditional expectation operator 
$L^2\{F(w|x)\}\rightarrow L^2\{F(\lambda |x)\}$, with $K_xh = E[H(w) | \lambda, x]$ for $H \in L^2\{F(w|x)\}$, and let $ (\tau_{x,n}, \varphi_{x,n}, \psi_{x,n})_{n=1}^\infty $ denote a singular value decomposition of $K_x$. Given the following regularity conditions:
\begin{enumerate}
    \item $\iint f(w | \lambda, x)f(\lambda | w, x)dwd\lambda < \infty$
    \item $\int f^2(y | \lambda, x)f(\lambda | x)d\lambda < \infty$
    \item $\sum_{n=1}^\infty| \langle f(y | \lambda, x), \psi_x,n \rangle |^2 < \infty$
\end{enumerate}
Then Picard’s theorem implies the existence of the required function $H^t$ in Theorem~\ref{theorem: for a single time point}.

\section{Definitions of Synthetic Experiments} \label{section: appendix simulation exp}
Our synthetic experiments are constructed such that the unobserved latent $w$ experiences a distribution shift after the intervention, leading to bias as defined in Equation~\ref{equation: bias}. 

\begin{equation}
     \epsilon \sim \mathcal{N}(0,1), u \sim \mathcal{N}(1,1) 
 \end{equation}
\begin{equation}
     Y = 
    \begin{cases}
    a u + b w + \epsilon, &\text{in } T_{pre} \text{ , } w \sim Bin(1/2) \\
    a u + b w + 2 I + \epsilon,&\text{in } T_{post} \text{ , } w \sim Bin(1) \\
\end{cases}
    \end{equation}
    
    \begin{equation}
Z = 
    \begin{cases}
    d w + \epsilon,& \text{in } T_{pre} \text{ , } w \sim Bin(1/2) \\
    d w + \epsilon,& \text{in } T_{post} \text{ , } w \sim Bin(1) \\
\end{cases}
 \end{equation}

\begin{equation}
    X =
    \begin{cases}
    cu + \epsilon,& \text{in } T_{pre} \text{ , } \epsilon \sim \mathcal{N}(0,1) \\
    cu + \epsilon,& \text{in } T_{post} \text{ , }  \epsilon \sim \mathcal{N}(0.5,1) \\
\end{cases}
\end{equation}

\begin{equation}
    I =
    \begin{cases}
    0,& \text{in } T_{pre} \\
    Bin(sigmoid(u))+ \epsilon,& \text{in } T_{post} \\
\end{cases}\\
    % \end{aligned}
    \end{equation}
    % \vspace{-5mm}

Given this data generation process, we have:
%  $$
 
%  $$
 
%  \begin{multicols}{2}
%     \begin{equation}

$$
    \textbf{ Synthetic Control } \E(Y) = \frac{a}{c}\E(X) + \frac{b}{d}\E(Z), 
\quad
      \textbf{True Bias} = |b\{ \E_{pre}(w) - \E_{post}(w) \}|,\\
 $$
 $$
 \textbf{ Proxies Bias  } =  |\frac{b}{d}\{ \E_{pre}(Z) - \E_{post}(Z) \}|,
\quad
    \textbf{Our bound on bias} = |\frac{a}{c}\{ \E_{pre}(X) - \E_{post}(X) \}|
   $$
    
Following from the above equation for the Synthetic Control $\E(Y)$, our bound on the bias holds if the following conditions hold: (1) $\frac{a}{c} > \frac{b}{d} $, i.e the \textit{weighting} of the contribution of the mean of the unobserved proxies $Z$ is smaller than of the mean of the observed proxies $X$. (2) The change in the \textit{mean} of proxies $X$ is bigger than the change in the \textit{mean} of unobserved proxies $Z$.

Given access to the unobserved $w$ and its proxies $Z$ through this simulation, we can validate the bounds. Setting $(a,b,c,d)=(1,0.5,1,0.5)$, Figure \ref{fig:simulation} showcases the aforementioned conditions in our synthetic setting under a valid and a invalid bound scenario. For Scenario (a), our maximum mean change in proxies $X$ is 0.48 (not exactly 0.5 due to noise terms), $N$ is 1 and the OLS coefficient is 1.47, such that the $ \text{bias} \leq 1 \times  1.47 \times  0.48 = 0.71$.
As expected, the bias (red) is captured by the bounds as both conditions are fulfilled, see bottom graph of Figure \ref{fig:simulation} (a). 

For Scenario (b), if we change the noise term on X in the post intervention stage to $\epsilon \sim \mathcal{N}(0.1,1)$, causing a violation of the second condition, we have a mean change of 0.08 (not 0.1 due to noise), leading to a $\text{bias} \leq 1 \times  1.47 \times  0.08 = 0.12$.
Hence, the bounds are smaller, but more importantly also invalid as they do not contain the true bias. 

Having chosen a simple example for effective exposition, we would like to emphasize that the validity (and invalidity) of the bounds in these scenarios naturally extend to more complex scenarios with higher number of latents $u$ and proxies $X$.

%\section{Real World Examples: German Reunification and California Tobacco Tax}

\end{document}